\documentclass[11pt]{article} %

\usepackage[margin=1in]{geometry}
\usepackage{natbib}
\usepackage{amsthm, amsmath, amssymb}

\usepackage{hyperref}
\usepackage{cleveref}

\newtheorem{theorem}{Theorem}
\newtheorem*{theorem*}{Theorem}
\newtheorem{proposition}[theorem]{Proposition}
\newtheorem*{proposition*}{Proposition}

\newtheorem*{claim*}{Claim}

\newtheorem*{observation*}{Observation}

\newtheorem{lemma}[theorem]{Lemma}
\newtheorem*{lemma*}{Lemma}

\newtheorem*{exercise*}{Exercise}
\newtheorem{corollary}[theorem]{Corollary}

\theoremstyle{definition}
\newtheorem{definition}[theorem]{Definition}

\title{
Sketching Intersection Profiles: \\
A Simple Proof and Three Applications
}

\usepackage{xcolor}         %
\usepackage{graphicx} %
\usepackage{amsmath}
\usepackage{amssymb}
\usepackage{nicefrac}
\usepackage{enumitem}
\usepackage{xspace}

\usepackage{thm-restate}

\DeclareMathOperator{\cut}{cut}

\DeclareMathOperator*{\argmax}{arg\,max}
\DeclareMathOperator*{\E}{E}

\DeclareMathOperator{\Dirichlet}{Dirichlet}
\DeclareMathOperator{\Beta}{Beta}

\renewcommand{\epsilon}{\varepsilon}

\newcommand{\norm}[1]{\left\|#1\right\|}

\newcommand{\IP}[1]{{F_{#1}}}

\newenvironment{proofof}[1]{%
  \begin{proof}[\noindent\textbf{\emph{Proof of #1}}]\ }{%
  \end{proof}%
}

\definecolor{Gred}{RGB}{219, 50, 54}
\definecolor{Ggreen}{RGB}{60, 186, 84}
\definecolor{Gblue}{RGB}{72, 133, 237}
\definecolor{Gyellow}{RGB}{204, 204, 0}
\definecolor{Gbrown}{RGB}{135, 64, 1}
\providecommand{\Comments}{1}

\ifnum\Comments=1
\usepackage[colorinlistoftodos,prependcaption,textsize=scriptsize]{todonotes}
\else
\usepackage[disable]{todonotes}
\fi
\newcommand{\mytodo}[1]{\ifnum\Comments=1{#1}\fi}

\author{Flavio Chierichetti\thanks{Reddit, San Francisco. \texttt{flavio.chierichetti@gmail.com}}
\hspace*{1.1cm}
Mirko Giacchini\thanks{Sapienza University of  Rome. \texttt{\{giacchini, ale, tani\}@di.uniroma1.it}}
\hspace*{0.9cm}
Ravi Kumar\thanks{Google, Mountain View.  \texttt{\{ravi.k53, atomkins\}@gmail.com}}
\\
Alessandro Panconesi\footnotemark[2]
\hspace*{1.2cm}
Erasmo Tani\footnotemark[2]
\hspace*{0.7cm}
Andrew Tomkins\footnotemark[3]
}

\date{}

\begin{document}

\maketitle

\begin{abstract}

In this work we settle the complexity of three sketching problems. (i) We show that sketching vertex neighborhood sizes in graphs requires $\Omega(n^2)$ bits, standing in sharp contrast to the $\tilde{O}(n)$ complexity of sketching edge cuts. (ii) We obtain tight lower and upper bounds of $\tilde{\Theta}(n^2)$ for sketching coverage functions with additive and multiplicative errors. (iii) We prove an $\Omega(n^2)$ lower bound for sketching Random Utility Models  under the $\ell_\infty$-norm, improving upon the previous $\Omega(n \log n)$ bound and matching a known upper bound to within logarithmic factors. 

These bounds are obtained through a connection with the problem of sketching the \emph{intersection profile} of a distribution $D$ on $2^{[n]}$. Specifically, we seek a succinct data structure that, for any query set $S \subseteq [n]$, approximates the quantity $\Pr_{T \sim D}[T \cap S \neq \varnothing]$ to within a small constant additive error.
One can obtain lower bounds for this latter problem directly from known results about the \emph{itemset frequency estimation problem} in databases for which tight bounds are known. 
As an additional contribution, we also provide an alternative  proof for the intersection profile sketching lower bound, in the setting in which the accuracy parameter is constant. This proof relies solely on elementary probability avoiding the heavier machinery used in previous proofs.
\end{abstract}

\thispagestyle{empty}
\setcounter{page}{0}
\newpage

\section{Introduction}
Sketching high-dimensional objects into succinct data structures is a fundamental primitive in algorithm design, enabling efficient processing of massive datasets in learning and optimization.  
A central theme in this area is to study the space complexity of the data structure: how many bits are necessary and sufficient to (approximately) answer queries about the underlying object? 

The main contribution of this work is to prove simple sketching lower bounds for three problems: (i) vertex neighborhoods of a graph, (ii)  coverage functions, and (iii) random utility models. All the lower bounds are also matched (up to polylog factors) by easy upper bounds.

These lower bounds follow from space complexity lower bounds for
approximating the intersection profiles of distributions over subsets of a universe $[n] = \{1, \dots, n\}$. Given a distribution $D$ over subsets of $[n]$, its \emph{intersection profile} is the function defined by $F_{D}(S) = \Pr_{T \sim D}[T \cap S \neq \varnothing]$, for $S\subseteq [n]$. The goal is to construct a succinct data structure, using which, for any query set $S$, one can compute $F_{D}(S)$ within a small additive error. This problem can equivalently be described in terms of the \emph{itemset frequency estimation} problem.  Here, one is given a binary matrix (i.e., a database) of $n$ columns (i.e., attributes) and $R$ rows. An itemset is a set $T\subseteq [n]$ of attributes and its frequency is the fraction of rows having a 1 in all columns $T$. Since $T\cap S= \varnothing$ if and only if $T \subseteq [n]\setminus S$, an additive sketch for frequency estimation translates into one for intersection profiles and vice-versa.\footnote{Note that in intersection profiles we allow the distribution $D$ to be arbitrary, and not necessarily uniform as in itemset frequency estimation, but assuming that $D$ is uniform over its support incurs only an extra additive error of $\epsilon$ in its estimate (see \Cref{thm:upper-bound-intersection-profile}).} Therefore, using known bounds for sketching itemset frequencies \citep{lmtu16}, we immediately obtain lower bounds for sketching intersection profiles.

In addition, we provide an alternative proof for a special case of the sketching lower bound of itemset frequencies using an elementary probabilistic packing argument, which is arguably simpler than the machinery employed to obtain the tight bounds of \citet{lmtu16}.

\subsection{A Simple Proof for Intersection Profiles}

We prove that any sketch that approximates the intersection profile of an arbitrary distribution over subsets of $[n]$ to within a constant additive error needs $\Omega(n^2)$ bits. Here, the error is measured as the maximum discrepancy between $F_D(S)$ and its estimate, over the choice of the set $S$ (the $\ell_\infty$-norm of the $F_D$ vector). 

This is a special case of the itemset frequency estimation result of \citet{lmtu16} in that we only allow constant $\epsilon$ and we focus the problem of estimating all subsets rather than all subsets of size $k$.  To compare, \citet{lmtu16} obtain an $\Omega \left(\frac{n\cdot k}{\epsilon^2} \log \frac{n}{k} \right)$ lower bound for estimating all subsets of size $k$ within error of $\epsilon$.

Our lower bound is established via a packing argument involving a family of distributions defined by random binary matrices, which we show are pairwise far apart in the $\ell_\infty$-distance. The technique of proving lower bounds via packing arguments is standard in sketching, but our construction involving the embedding of random matrices to separate intersection profiles is, to the best of our knowledge, novel to this setting. %

The technical difficulty in showing the lower bound arises from the space of distributions over subsets, which is highly non-uniform.  Indeed, for a distribution that is sampled uniformly at random from the simplex, the probability of intersecting any fixed $S$ is concentrated around $1 - 2^{-|S|}$ (see \Cref{sec:typical-set}). Consequently, for most instances, a constant number of bits suffices to approximate the profile, since, on input $S$ one can simply return $1 - 2^{-|S|}$ as an estimate for ${F}_D(S)$ and it will be correct with high probability. However, we show that worst-case instances define a geometry that is much harder to capture.  These instances occupy thin, low-measure regions of the parameter space, and their geometric ``thinness,'' along with the non-uniform structure of the space, are the main obstacles in constructing the lower bound. In particular, standard volumetric arguments fail, as random instances do not capture the true geometry of the underlying space.

\subsection{Applications}
We use the lower bound for sketching intersection profiles to resolve open questions and tighten bounds for three key applications:

\begin{enumerate}[label=(\roman*), wide=0pt]
\item \emph{Vertex neighborhoods.} It is known that one can sketch the edge cut function of a graph within multiplicative error $\varepsilon$ with $O(n\log n /\varepsilon^2)$ bits, which is optimal \citep{ckst19}. It is natural to ask if a similar result is possible for \emph{vertex} boundaries or equivalently, neighborhood cardinalities. We define the vertex neighborhood sketching problem as the task of generating sketches to approximate the size of the neighborhood for any subset $S \subseteq V$. 
We prove that, unlike edge cuts, vertex boundaries fundamentally \emph{resist}  sparsification: any data structure preserving neighborhood sizes up to an additive $0.01 n$ error requires $\Omega(n^2)$ bits, meaning that exactly representing the graph, using, e.g.,\ its adjacency matrix, is optimal. We show our result holds both for the problem of sketching vertex neighborhoods, which include all neighbors of elements of $S$, and for that of sketching outer vertex neighborhoods (a.k.a., outer vertex boundaries), which only count the neighbors that are \textit{not} part of $S$ itself.

\item \emph{Coverage functions.} Coverage functions are an important subclass of submodular functions used extensively in optimization and data mining. A coverage function $f(A) = \frac{1}{m} |\bigcup_{i \in A} S_i|$ measures the (normalized) size of the union of subsets of $[m]$ indexed by $A$. We show that our lower bound for sketching intersection profiles  extends to coverage functions, implying that $\Omega(n^2)$ bits are necessary for constant additive error sketches (and hence for $(1\pm\epsilon)$-multiplicative error as well), which is tight. 
Furthermore, for $(1 \pm \epsilon)$-multiplicative sketching, we extend an algorithm of \cite{bem17} to provide a sketching using $\tilde{O}(n^2/\epsilon^2)$ bits, matching our lower bound up to logarithmic factors.

\item \emph{Random Utility Models (RUMs)}. RUMs are distributions over permutations of a universe of $n$ elements that are used to model preferences in economics and machine learning (see, e.g.,\ \citep{t09}). For a non-empty subset $S\subseteq [n]$ (also called a ``slate'') and a permutation $\pi$ sampled from the RUM, the winner in $S$ is the item ranked highest in $S$ according to $\pi$. For each $i\in S\subseteq [n]$, the \emph{winning probability} of $i$ in $S$ is the probability that $i$ wins in $S$ under a permutation drawn from the RUM. Therefore, for each slate $S$, a RUM induces a \emph{winning distribution} over the items of $S$. 

To sketch the winning distributions of a RUM, we seek a data structure that, for each slate $S$, returns an estimate of the winning distribution over $S$. The error in the sketch is taken to be the worst-case distance (over the choice of $S$) between the true winning distribution and the estimate obtained from the sketch. This distance can either be measured in $\ell_1$, as twice the total variation distance, or in $\ell_\infty$, as the maximum error over the individual probability of any item in $S$.

\cite{chierichetti2021light} showed that sketching the winning distributions of a RUM under the $\ell_1$-distance requires $\Omega(n^2)$ bits, and that $O(n^2 \log n)$ bits suffice. However, for the $\ell_\infty$-distance---which seeks to estimate the winning probability of specific items in a slate---the best known lower bound was the trivial $\Omega(n \log n)$ bits, which are needed to even sketch a single permutation.  We improve this to $\Omega(n^2)$ via a reduction from our sketching intersection profiles problem. 

Interestingly, our lower bound also applies to the simpler, related problem of sketching all winning probabilities of a \emph{single} item in the universe, for which this bound is tight.
\end{enumerate}
\section{Related Work}

\paragraph{Sketching itemset frequencies.} Given a binary matrix of $n$ columns and $R$ rows, an itemset is defined as a set $T\subseteq[n]$ and its frequency is the fraction of rows where all the columns $T$ contain a 1. \citet{lmtu16} showed that any sketch which is able to additively approximate the frequency of any itemset of size $k$ must require $\Omega(\frac{nk}{\epsilon^2} \log(n/k))$ bits. For the case where we want to estimate all itemsets, this translates into a $\Omega(n^2/\epsilon^2)$ lower bound. This lower bound immediately implies an analogous lower bound for the sketching of intersection profiles of power set distributions. For convenience, we used this second abstraction to obtain lower bounds for the other applications. 

The $\Omega(\frac{nk}{\epsilon^2} \log(n/k))$ lower bound of \citet{lmtu16} builds on spectral random matrix theory, and they also have a simpler bound of $\Omega(nk\log(n/k))$ which uses a VC-dimension argument together with an error-correcting code. Our argument provides an arguably simpler proof, albeit for a weaker version of the lower bound.

\paragraph{Sketching edge cuts.}
A long line of work has focused on the problem of producing small representations of graphs that multiplicatively approximate the edge-cut function (i.e., the number of edges cut) for all input cuts~\citep{k93,bk96,ss08, st11,alo15,ls18}.
\cite{bss12} obtained sketches for edge cuts using $O({n\log n\over \varepsilon^2})$ bits; this is also known to be optimal \citep{ckst19}.
It is a natural question to seek similar-sized sketches for the vertex boundary function.
We will answer this question in the negative, showing that $\Omega(n^2)$ bits are required for a multiplicative error. 

\paragraph{Sketching coverage functions.} Coverage functions are a special case of submodular functions, which have attracted great interest in learning theory \citep{b13,fk14,ch15,yhydyys21} and optimization \citep{bh18}. Moreover, optimization problems on coverage functions entail classical algorithmic problems such as max-cover and set-cover as special cases. From the sketching perspective, \citet{bdfknr12} showed that coverage functions can be sketched to within a multiplicative $(1\pm \epsilon)$ error with bit complexity $\tilde{O}(n^3 /\epsilon^2)$, and \citet{bem17} provided a sketch with $\tilde{O}(n k)$ bits for the special case of $k$-cover. It is possible to extend the latter to obtain a sketch with $\tilde{O}(n^2)$ bits for general coverage functions. Randomized sketches minimizing the expected squared error have also been studied \citep{yz19}---we do not consider this error in our paper. To the best of our knowledge, no non-trivial lower bound for sketching coverage functions was known in the literature.

\paragraph{Relationship to $\ell_1$-distance bounds for Random Utility Models.}

Our alternative lower bound proof shares the high-level packing argument used for the 
$\ell_1$-distance sketching lower bound
by~\citet{chierichetti2021light} for Random Utility Models—generating $\exp(n^2)$ instances and
proving they are pairwise far apart—there is an important technical difference in how the separation is shown. The separation for the $\ell_1$ result uses a fixed class of $\Theta(n)$ \emph{test sets}, guaranteeing that any pair of RUMs will behave differently on at least one of these fixed subsets of the universe. In contrast, for the $\ell_\infty$-distance, fixing $m$ test sets \emph{a priori} allows the distribution to be
approximated with only $O(n \log(n)\log(m))$ bits---which precludes an
$\Omega(n^2)$ lower bound unless $m$ is exponentially large. Hence, our proof requires choosing the $\Theta(n)$ test sets adaptively based on the structure of the sampled instances to guarantee the necessary separation probability.

\section{Background: Sketching and Intersection Profiles}

We begin by reviewing the notion of sketching. In the typical sketching setting, one is interested in producing a succinct representation of an object from a fixed class in a way that allows them to approximately answer certain questions of interests about the object by looking solely at this representation. As is customary, given sets $A$ and $B$ we will denote by $A^B$ the collection of functions from $B$ to $A$. We will reserve the notation $2^A$ for the power set of $A$: the collection of all subsets of $A$. For a set $A$ and integer $k$, we denote with $\binom{A}{k}$ all the subsets of $A$ of cardinality $k$. For a proposition $P$, we will let $[P]$ be 1 if $P$ is true and 0 otherwise.

\begin{definition}[Sketch]
    Let $\mathcal{F} \subseteq \mathbb{R}^{Q}$ be a class of real-valued functions on a set $Q$. A \emph{sketch} for $\mathcal{F}$ is a pair $(\phi,\psi)$ of functions with $\phi:\mathcal{F} \to \{0,1\}^{\kappa}$ and $\psi:\{0,1\}^{\kappa}\times Q \to \mathbb{R}$, where $\kappa \in \mathbb{N}$ is a parameter called the \emph{bit complexity} of the sketch.  A sketch $(\phi,\psi)$ is \emph{$\varepsilon$-multiplicative} if it satisfies:
    \[
       \forall f\in \mathcal{F}, \forall q\in Q: \;  |\psi(\phi(f),q)-f(q)|\leq \varepsilon \cdot |f(q)|.
    \]  
    A sketch $(\phi,\psi)$ is \emph{$\varepsilon$-additive} if it satisfies:
    \[
        \forall f\in \mathcal{F}, \forall q\in Q: \; |\psi(\phi(f),q)-f(q)|\leq \varepsilon.
    \]
\end{definition}
Intuitively, given a function $f$, one wants to produce a short bit string $\phi(f)$ that allows them to approximate the value of $f(q)$ for every possible element $q\in Q$ (sometimes referred to as a \emph{query}).

We will refer to a probability distribution supported on $2^{[n]}$ as a \emph{power-set distribution} on $[n]$, or simply a \emph{power-set distribution}, and we will denote by $\mathcal{D}_n$ the collection of all power-set distributions on $[n]$. In this paper, we obtain sketching lower bounds leveraging on results for intersection profiles of power-set distributions, defined as follows. 
\begin{definition}[Intersection Profile]
For a power-set distribution $D \in \mathcal{D}_n$, its \emph{intersection profile} is the function $\IP{D} \in [0,1]^{2^{[n]}}$ defined as:
\[
\IP{D}(S) = \Pr_{T\sim D}[T\cap S \neq \varnothing].
\]
\end{definition}
Let us also denote with $F_{D,k}\in[0,1]^{\binom{[n]}{k}}$ the restriction of $F_D$ to subsets of size $k$. We will leverage on the following results of \citet{lmtu16}, which concern the problem of sketching the class of intersection profiles of power-set distributions: $\mathcal{F}_n = \{F_D(\cdot) \mid D\in \mathcal{D}_n\}$ and $\mathcal{F}_{n,k} = \{F_{D,k}(\cdot) \mid D\in \mathcal{D}_n\}$.

For convenience, for two power-set distributions $D, D'$ and $k\geq 1$, we define the distance $d_k(D,D')$ based on the $\ell_\infty$-distance between their intersection profiles:
\[
    d_k(D, D):= \norm{F_{D,k} - F_{D',k}}_\infty = \max_{S\in \binom{[n]}{k}} |F_D(S) - F_{D'}(S)|.
\]
We also define: 
\[
d_\infty(D, D') := \norm{\IP{D} - \IP{D'}}_{\infty} = \max_{S \subseteq [n]} |\IP{D}(S) - \IP{D'}(S)|.
\]
From the upper bound side we have:
\begin{theorem}[Lemma 9 of \citet{lmtu16}, paraphrased]\label{thm:upper-bound-intersection-profile} For each $\epsilon > 0$ and for every power-set distribution $D$, there is a power-set distribution $D'$ that is uniform\footnote{The support of $D'$ can be a multiset.} over a support of size $O\left(n/\epsilon^{2}\right)$ (resp., $O(\frac{1}{\epsilon^2} \cdot \log\binom{n}{k})$) and such that $d_\infty(D, D') \leq \epsilon$ (resp., $d_{k}(D,D')\leq \epsilon$).

Since each set requires $O(n)$ bit to be stored, we have that there exists an $\epsilon$-additive sketch for $\mathcal{F}_n$ (resp., $\mathcal{F}_{n,k}$) of bit complexity $\kappa=O(n^2/\epsilon^2)$ (resp., $\kappa=O(n\cdot k\cdot \log(\frac{n}{k})/\epsilon^2)$).
\end{theorem}

Our main tool to prove lower bounds for the three applications is the following lower bound for sketching intersection profiles:

\begin{theorem}[Theorem 16 of \citet{lmtu16}, paraphrased]\label{thm:set-intersection-sketching-lower-bound}
    Let $3\leq k\leq n/2$ and $\epsilon \in \left(\frac{1}{\sqrt{n}}, 1\right)$. Then, any $\epsilon$-additive sketch for $\mathcal{F}_{n,k}$ has bit complexity $\Omega(n \cdot k \cdot \log(\frac{n}{k}) / \epsilon^2)$.

    By choosing $k=\Theta(n)$, this implies a lower bound of $\Omega(n^2/\epsilon^2)$ for any $\epsilon$-additive sketch of $\mathcal{F}_{n}$.
\end{theorem}

\section{A Simple Lower Bound for Sketching $\mathcal{F}_n$}\label{sec:main-results}

In this section, we provide a simple alternative proof for the lower bound on the bit complexity of additive sketches for $\mathcal{F}_n$ in the special case where $\epsilon$ is a constant. Specifically, we give a self-contained proof of the following:

\begin{theorem}\label{thm:simple-lower-bound}
    Any $0.05$-additive sketch for $\mathcal{F}_n$ has bit complexity $\Omega(n^2)$.
\end{theorem}

The main idea required to prove this result is to construct $2^{\Omega(n^2)}$ power-set distributions that are pairwise far apart. In particular, the key technical step is the following result.

\begin{theorem}\label{thm:set-intersection-ell-infty-packing}
    For every $n$, there exists a collection $D_1, \ldots, D_K$ of $K = K(n) = 2^{\Omega(n^2)}$ power-set distributions such that 
    $
        d_\infty (D_i , D_j) > 0.1
    $
    for each $\{i, j\} \in \binom{[K]}2$.
\end{theorem}
Before we proceed with its proof, we note that \Cref{thm:simple-lower-bound} follows directly from \Cref{thm:set-intersection-ell-infty-packing} by a standard argument, which we now give for completeness.

\begin{proof}[Proof that \Cref{thm:simple-lower-bound} follows from \Cref{thm:set-intersection-ell-infty-packing}]
    By contradiction, suppose there is a $0.05$-additive sketch $(\phi, \psi)$ for $\mathcal{F}_n$ with bit complexity $\kappa < \log_2 K$, where $K$ is chosen as in the statement of \Cref{thm:set-intersection-ell-infty-packing}.  Now, consider a collection of $K$ set intersection instances $D_1, \ldots, D_K$ with the property guaranteed by \Cref{thm:set-intersection-ell-infty-packing}.  By the pigeonhole principle, there exist $D_i$ and $D_j$ with $i \neq j$ such that $\phi(D_i) = \phi(D_j)$.  Also, since $d_\infty(D_i,D_j) > 0.1$ there exists some subset $T\subseteq [n]$ such that $|{F_{D_i}}(T) - {F_{D_j}}(T)| >0.1.$
    
    On the other hand, by the triangle inequality, 
    \begin{align*}
        |{F_{D_i}}(T) - {F_{D_j}}(T)| 
        &\leq |{F_{D_i}}(T) -\psi(\phi(D_i),T)| +|\psi(\phi(D_i),T) - {F_{D_j}}(T)| \\
        &= | F_{D_i}(T) -\psi(\phi(D_i),T)| +|\psi(\phi(D_j),T) - {F_{D_j}}(T)| \\
        &\leq 0.05 + 0.05 = 0.1, \mbox{ a contradiction.}  
    \end{align*}
    Hence, every $0.05$-additive sketch for $\mathcal{F}_n$ must have bit complexity $\kappa \geq \log_2 K = \Omega(n^2)$, completing the proof.
\end{proof}

We now prove Theorem~\ref{thm:set-intersection-ell-infty-packing}. We do so via a probabilistic packing argument, constructing a family of distributions that are pairwise far apart in the $d_\infty(\cdot, \cdot)$-metric. The construction embeds a random binary matrix into the problem by partitioning the universe items into one of two types: \emph{high-degree} items and \emph{low-degree} items. High-degree items appear in roughly half the sets of our power-set distribution (i.e., they have high degree in the items-sets bipartite graph) and encode the high-entropy random bits of the matrix. Low-degree items appear in exactly one set each (i.e., they have low degree in the items-sets bipartite graph), and help identify the set uniquely. A distribution consisting only of high-degree items could be sketched with few bits because, with high probability, for all sets $S \subseteq [n]$, the set $S$ will have empty intersection with a sampled set with probability $\approx 2^{-|S|}$. On the other hand, a distribution consisting only of low-degree items could be sketched using $O(n \log n)$ bits, since each set could be represented exactly.

To prove our $\Omega(n^2)$ lower bound, we leverage the interplay of the high-degree and low-degree items.
The two types of items together make it possible to create $\Theta(n)$ queries such that, for the sketch to answer each of those queries, it has to hold a very good approximation of the neighborhoods of the high-degree items in the items-sets bipartite graph.

\begin{figure}[t]
\centering
\[
\left[M \;|\;  I_{m}\right] = \left(
\begin{array}{cccc|cccc}
M_{1,1} & M_{1,2} & \cdots & M_{1,m} & 1 & 0  & \cdots & 0  \\
M_{2,1} & M_{2,2} & \cdots & M_{2,m} & 0 & 1 & \cdots & 0 \\
\vdots  & \vdots & \ddots & \vdots & \vdots & \vdots & \ddots  & \vdots \\
M_{m,1} & M_{m,2} & \cdots & M_{m,m} & 0 & 0 & \cdots & 1
\end{array}
\right)
\]
\[
    [I_m|\mathbf{1}\mathbf{1}^\top - M^{\top}] = \left(
\begin{array}{cccc|cccc}
1 & 0  & \cdots & 0 & 1-M_{1,1} & 1-M_{2,1} & \cdots & 1-M_{m,1} \\
0 & 1 & \cdots & 0 & 1-M_{1,2} & 1-M_{2,2} & \cdots & 1-M_{2,m} \\
\vdots  & \vdots & \ddots & \vdots & \vdots & \vdots & \ddots  & \vdots \\
0 & 0 & \cdots & 1 & 1-M_{1,m} & 1-M_{2,m} & \cdots & 1-M_{m,m}
\end{array}
\right)
\]

\caption{First, the $m \times n$ matrix $\left[M \;|\;  I_{m}\right]$ with $M_{i,j} \sim \text{Bernoulli}(1/2)$. Each row $i$ of this matrix is the indicator vector of the set $S_i^{(M)}$. The power-set distribution $D_M$ associated with the matrix $M$ selects a subset by sampling a row of this matrix uniformly at random, and interpreting the row as the indicator of the corresponding subset of $[n] = [2m]$. Below that, the matrix $[I_m|\mathbf{1}\mathbf{1}^\top - M^{\top}]$. Each row $j$ of this matrix, is the indicator of the query set $T_j^{(M)}$}
\label{fig:Bernoulli}
\end{figure}

\begin{proof}[Proof of \Cref{thm:set-intersection-ell-infty-packing}]
For simplicity, we assume $n$ is even; let $m = n/2$.  Each distribution is associated with a random binary matrix $M \in \{0, 1\}^{m \times m}$, where each entry $M_{i,j} \sim \mathrm{Bernoulli}(1/2)$, and is independent of all other entries. The matrix $M$ is used to define two quantities: (i) the distribution $D_M$ and (ii) a family of queries.  Formally, for any given $M \in \{0,1\}^{m\times m}$, we let $D_M$ be the uniform distribution on $m$ sets $S_1^{(M)} \ldots, S_m^{(M)} \subseteq [n]$ that, intuitively, correspond to the rows of $M$ paired with an identity matrix.  Specifically, for each $i \in [m]$, let
\[ S_i^{(M)} = \{ j\in [m] \mid M_{i,j} = 1 \} \cup \{ m + i \}.
\]
The query family $T_1^{(M)}, \ldots, T_m^{(M)} \subseteq [n]$ is, intuitively, given by the columns of an identity matrix paired with the complement of $M$.  Specifically, for each $j \in [m]$, let
\[
    T_j^{(M)} = \{j\} \cup \{m + i \mid i \in [m] \mbox{ and } M_{i,j} = 0 \}.
\]
We show an example in \Cref{fig:Bernoulli}. We start by showing some key properties of this construction.  

\begin{lemma}\label{lem:AB}
Let $A$ and $B$ be two (not necessarily distinct) matrices in $\{0,1\}^{m\times m}$, and let $S_i^{(A)}$ and $T_j^{(B)}$ be defined as above. For any $i, j \in [m]$, 
$S_i^{(A)} \cap T_j^{(B)} \neq \varnothing$ if and only if 
$A_{i,j} = 1$ or $B_{i,j} = 0$.
\end{lemma}
\begin{proofof}{Lemma~\ref{lem:AB}}
Note that the only element of $S_i^{(A)} \cap \{m+1, \dots , 2m\} = m+i$ and $T_j^{(B)} \cap \{1, \dots , m\} = j$ and hence $S_i^{(A)} \cap T_j^{(B)}  \subseteq \{j, m+i\}$. In particular, $S_i^{(A)} \cap T_j^{(B)} $ is non-empty if and only if either $j \in S_i^{(A)}$ or $m+i \in T_j^{(B)}$. By construction, this happens if and only if $A_{i,j} = 1$ or $B_{i,j} = 0$.
\end{proofof}
This leads to two useful consequences.
\begin{corollary}
\label{cor:AB1}
Let $D_A$ be any power-set distribution induced by a binary matrix $A$ as above. Then, for any query $T_j^{(A)}$, we have
$F_{D_A}(T_j^{(A)}) = 1$. 
\end{corollary}
\begin{proofof}{Corollary~\ref{cor:AB1}}
Using \Cref{lem:AB} with $B = A$, we obtain $S_i^{(A)} \cap T_j^{(A)} \neq \varnothing$ always.  Hence, 
\[
        F_{D_A}(T_j^{(A)}) = \frac{1}{m} \sum_{i=1}^m \underset{\text{binary indicator}}{\underbrace{\left[ S_i^{(A)} \cap T_j^{(A)} \neq \varnothing \right]}} = 1.
\qedhere
\]
\end{proofof}
\begin{corollary}
\label{cor:AB2}
For any $i, j \in [m]$:
\[
    \Pr_{A, B}[S_i^{(A)} \cap T_j^{(B)} \neq \varnothing] = 3/4,
\]
where the entries of $A$ and $B$ are sampled independently from $\mathrm{Bernoulli}(1/2)$.
\end{corollary}
\begin{proofof}{Corollary~\ref{cor:AB2}}
Using \Cref{lem:AB} and the independence of $A$ and $B$, 
\[
\Pr_{A, B}[S_i^{(A)} \cap T_j^{(B)} = \varnothing] = \Pr_{A, B}[A_{i,j} = 0 \mbox{ and } B_{i,j} = 1] = \frac{1}{2} \cdot \frac{1}{2} = \frac{1}{4}.
\qedhere
\]
\end{proofof}
Let $D_A$ and $D_B$ be two independent instances defined by matrices $A$ and $B$.  Let $j \in [m]$ be a fixed column and let $Y_j$ be the random variable representing the value of the intersection profile of $D_A$ on the query $T_j^{(B)}$ (note that this random variable is measurable in the $\sigma$-algebra generated by the $j$th columns of $A$ and $B$, and its value is determined once $A$ and $B$ are fixed):
    \[
        Y_j = \IP{D_A}(T_j^{(B)}) = \frac{1}{m} \sum_{i=1}^m \left[ S_i^{(A)} \cap T_j^{(B)} \neq \varnothing \right] = \frac{1}{m} \sum_{i=1}^m \left[ A_{i,j} = 1 \mbox{ or } B_{i,j} = 0 \right].
    \]
From this, $Y_j$ is the average of $m = n/2$ i.i.d. Bernoulli random variables of parameter $3/4$ (Corollary~\ref{cor:AB2}).  By a Hoeffding bound (see, e.g.,\ \cite[Theorem 1.1]{d09}), we have:
\[
        \Pr_{A, B}\left[Y_j \ge 0.9 \right] = \Pr_{A, B} \left[ Y_j - \mathbb{E}[Y_j] \ge \frac{3}{20} \right] \leq \exp\left(-\frac{9n}{400} \right).
\]
Note that if $Y_j < 0.9$, then, by \Cref{cor:AB1}, the distance between the intersection profiles on this specific query $T_j^{(B)}$ is:
\[
        |\IP{D_B}(T_j^{(B)}) - \IP{D_A}(T_j^{(B)})| = |1 - Y_j| > 0.1.
\]  
Now, the probability that the distance between $D_A$ and $D_B$ is less than 0.1 on all queries $T_1^{(B)}, \dots, T_m^{(B)}$ is given by:
    \[
        \Pr_{A, B}[d_\infty(D_A, D_B) \le 0.1] 
        \le \prod_{j=1}^m \Pr[Y_j \ge 0.9] 
        \le \prod_{j=1}^m \exp\left(-\frac{9n}{400} \right)
        = \exp \left(- \frac{9n^2}{800} \right),
    \]
where we are crucially using the fact that, by \Cref{lem:AB}, each variable $Y_j$ is a function of the $j$th column of $A$ and $B$ only, and hence all the $Y_j$'s are mutually independent.
    
Applying a union bound over all pairs in a collection of $K$ distributions 
$D_1, \ldots, D_K$ corresponding to matrices $M_1, \ldots, M_K$ sampled independently, we obtain: \[
        \Pr_{M_1,\dots, M_K }\left[\exists \{i,j\} \in \binom{[K]}2 \Bigm| d_\infty(D_{M_i},D_{M_j}) \leq 0.1\right] \leq \binom{K}{2} \exp \left(- \frac{9n^2}{800} \right) < 1,
    \]
    for $K \leq 2^{\left( \frac{9 \log_2 e}{1600} n^2\right)}$, completing the proof.
\end{proof}

\section{Application: Sketching Vertex Neighborhoods}

We now turn to applications of \Cref{thm:set-intersection-sketching-lower-bound}. The first application consists in proving lower bounds on sketching vertex neighborhoods. For any unweighted undirected graph $G=(V,E)$ and any $S \subseteq V$, the \emph{cut set} of $S$ is the set of edges with exactly one endpoint in $S$: $\cut(S) = \{e \mid e \in E \mbox{ and } |e \cap S| = 1\}$. It is known (\cite{bss12,ckst19}) that one can obtain an $\varepsilon$-multiplicative sketch for the function $|\cut(\cdot)|$ with bit complexity $\Theta(n\log n /\varepsilon^2)$ and that this bound is optimal.

Given an undirected graph $G=(V,E)$, for a subset $S \subseteq V$, we define $\Gamma_G(S)$ to be the set of \emph{neighbors} of $S$: $\Gamma_G(S) = \{w \mid w \in V \mbox{ and } \exists v \in S \text{ such that } \{v,w\}\in E\}$.  Observe that the adjacency matrix of the full graph, which requires $\Theta(n^2)$ bits, is a perfect (0-additive) sketch for $|\Gamma_G(\cdot)|$.  In contrast to the (edge-based) cut function, we show by using Theorem~\ref{thm:set-intersection-sketching-lower-bound} that this simple representation is asymptotically optimal.

We will actually show that the bit complexity is $\Omega(n^2)$ even if we restrict ourselves to only sketching bipartite graphs $G=(V_1 \cup V_2, E)$, and we only allow queries $S\subseteq V_1$. Specifically, let $\mathcal{B}_{n,m}$ be the set of all bipartite graphs with $|V_1|=n$ and $|V_2|=m$. %

\begin{theorem}\label{thm:vertexlb}
Let $3\leq k\leq n/2$ and $\epsilon \in \left(\frac{1}{\sqrt{n}}, 1\right)$. There exists $m=\Theta(k\log(\frac{n}{k})/\epsilon^2)$ such that the bit complexity of any $(\epsilon\cdot m)$-additive sketch for $\{|\Gamma_{G}(\cdot)| \mid G\in \mathcal{B}_{n,m}\}$ is $\Omega(n\cdot k \cdot \log(\frac{n}{k})/\epsilon^2)$, even if the domain of $\Gamma_G$ is restricted to $\binom{V_1}{k}$. 
\end{theorem}
\begin{proof}
    Let $D$ be any power-set distribution over $[n]$. By \Cref{thm:upper-bound-intersection-profile}, there exists a power-set distribution $D'$ over $[n]$ that is uniform over a support of $m=\Theta(k \log(\frac{n}{k})/\epsilon^2)$ (not necessarily distinct) sets and such that $|F_D(S) - F_{D'}(S)| \leq \epsilon$ for each $S\in \binom{[n]}{k}$.  

    Suppose there exists an $\epsilon m$-additive sketch $(\phi, \psi)$ for $\{|\Gamma_G(\cdot)|\}_{G\in \mathcal{B}_{n,m}}$ with bit complexity $\kappa$. Specifically, for any bipartite graph $G=(V_1 \cup V_2, E) \in \mathcal{B}_{n,m}$, let $\Delta_G(\cdot)=\psi(\phi(G), \cdot)$. Then, $\left||\Gamma_G(S)| - \Delta_G(S)\right| \leq \epsilon \cdot m$ for each $S\in \binom{V_1}{k}$.
    
    Let the support of $D'$ be the multiset $\mathcal{S} = \{S_1, \dots, S_m\}$, where $S_i \subseteq [n]$ for each $i$. We build the following bipartite graph $G = (V_1 \cup V_2, E)$, with $V_1=[n]$ and $V_2 = \mathcal{S}$. For each $x\in V_1$ and $S\in V_2$, we add the edge $\{x, S\}$ to $E$ if and only if $x\in S$. Observe that, for each $S\subseteq[n]$, it holds that: 
    \[
        |\Gamma_G(S)| = \sum_{i\in [m]} \left[S \cap S_i \neq \varnothing \right] = m \cdot F_{D'}(S).
    \]
    Then, for each $S\in \binom{[n]}{k}$, we have:
    \[
        \left| \frac{\Delta_G(S)}{m} - \frac{|\Gamma_G(S)|}{m} \right| \leq \frac{\epsilon \cdot m}{m} \leq \epsilon.
    \]
    Putting it all together, we have that for all $S\in\binom{[n]}{k}$:
    \begin{align*}
        \left| \frac{\Delta_G(S)}{m} - F_D(S) \right| &\leq \left| \frac{\Delta_G(S)}{m} - \frac{|\Gamma_G(S)|}{m} \right| + \left| \frac{|\Gamma_G(S)|}{m} - F_{D'}(S) \right| + \left| F_{D'}(S) -F_{D}(S) \right|  \leq 2\epsilon .
    \end{align*}

    This immediately implies a $2\epsilon$-additive sketch $(\phi, \psi(\cdot, \cdot)/N)$ for $\{F_D\}_{D\in\mathcal{D}_n}$ with bit complexity $\kappa + O(\log m)$. Therefore, by \Cref{thm:set-intersection-sketching-lower-bound}, the bit complexity of the sketch for $\{|\Gamma_G(\cdot)|\}_{G\in \mathcal{B}_{n,m}}$ must be $\kappa=\Omega(n\cdot m)$. 
\end{proof}

This result immediately yields a lower bound on the bit complexity of $\varepsilon$-multiplicative sketches of neighborhood cardinalities for arbitrary undirected graphs. In particular, let $\mathcal{G}_n$ be the set of all undirected graphs on $n$ vertices, we have the following:
\begin{corollary}
    Let $\alpha\in(0,1)$ be a constant. Then, any $\alpha$-multiplicative sketch for $\{|\Gamma_G(\cdot)| \mid G\in \mathcal{G}_n\}$ has bit complexity $\Omega(n^2)$.
\end{corollary}

Moreover, since our construction applies to sketching the neighborhood cardinalities of sets on one side of bipartite graphs, the lower bound applies to the problem of sketching outer neighborhood cardinalities. That is, given an undirected graph $G$ let:
\[
    \Gamma^{\text{out}}_G(S):= \Gamma_G(S)\setminus S = \{v \in V\setminus S \mid \exists u\in S, \{u,v\} \in E\},
\]
then, we have the following result.

\begin{corollary}
    Let $\alpha \in(0,1)$ be a constant. Then, any $\alpha$-multiplicative sketch for $\{|\Gamma^{\text{out}}_G(\cdot)| \mid G\in \mathcal{G}_n\}$ has bit complexity $\Omega(n^2)$.
\end{corollary}

\section{Applications: Sketches for Coverage Functions}\label{sec:coverage-func}

We now explore the second application of \Cref{thm:set-intersection-sketching-lower-bound}: sketching coverage functions. We recall the key definitions.
\begin{definition}[Normalized Coverage Function]
Given a universe $U$ of cardinality $m$ and a collection $S_1,\ldots, S_n \subseteq U$ of subsets, the \emph{normalized coverage function} $f: 2^{[n]} \rightarrow [0,1]$ associated with $S_1,\ldots, S_n$ and the universe $U$ is defined as
\[
f(A) = \frac{1}{m} \cdot \left|\bigcup_{i \in A} S_i\right|,
\]
where $A\subseteq [n]$. Given a collection of positive integer weights $W= \{w_{u}\}_{u\in U}$ each associated with an element of the universe $U$, the \emph{weighted coverage function} $f_W$ associated with $S_1, \dots S_n$ and the universe $U$ is defined as:
\[
    f_W(A) = {1\over w(U)} \cdot w\left(\bigcup_{i\in A} S_i\right)
\]
where, for any subset $B\subseteq U$, we define:
\[
    w(B) := \sum_{i\in B} w_i.
\]
\end{definition}

We consider the problem of sketching the family of all normalized coverage functions. Specifically, let $\mathcal{C}_{n,m}$ be the set of all normalized coverage functions with universe size $m$ and a collection of $n$ subsets. We also let $\mathcal{C}_{n,m,k}$ be the class of normalized coverage functions where we restrict the domain of the functions to $\binom{[n]}{k}$ rather than $2^{[n]}$. We show that the bounds obtained thanks to \Cref{thm:set-intersection-sketching-lower-bound} are tight both for additive error sketches and for multiplicative error sketches (in the latter case, up to polylog factors). 

\subsection{Additive-Error Sketch}\label{sec:coverage-func-additive}

\begin{theorem}\label{thm:coveragelb}
    Let $3\leq k\leq n/2$ and $\epsilon\in\left(\frac{1}{\sqrt{n}}, 1\right)$. Then, there exists $m=\Theta\left(\frac{k}{\epsilon^2} \log({n\over k})\right)$ such that any $\epsilon$-additive sketch for the class of normalized coverage functions $\mathcal{C}_{n,m,k}$ requires bit complexity $\Omega(n \cdot k \cdot \log(\frac{n}{k}) / \epsilon^2)$.
\end{theorem}
\begin{proof}
    Let $D\in \mathcal{D}_n$ be any power-set distribution over $[n]$. By \Cref{thm:upper-bound-intersection-profile}, there exists a power-set distribution $D'$ over $[n]$ that is uniform over a support of $m=\Theta(\frac{k}{\epsilon^2} \log(n/k))$ sets and such that $|F_D(S) - F_{D'}(S)| \leq \epsilon$ for each $S\in \binom{[n]}{k}$. Let $\mathcal{T}=\{T_1, \dots, T_m\}$ be the support of $D'$. 

    Consider now the coverage function $f$ having universe $[m]$ and sets $S_1, \dots, S_n$ such that: $S_i = \{j\in [m] \mid i\in T_j\}$ for each $i\in[n]$. Observe that for each $A \subseteq [n]$ we have:
    \[
        f(A) = \frac{1}{m}\cdot \left|\bigcup_{i\in A} S_i\right| = \frac{1}{m} \sum_{j\in[m]} [A \cap T_j \neq \varnothing] = F_{D'}(A).
    \]
    Suppose there exists an $\epsilon$-additive sketch $(\phi, \psi)$ for $\mathcal{C}_{n,m,k}$ with bit complexity $\kappa$. Specifically, let $\Delta_f(\cdot)=\psi(\phi(f), \cdot)$. For each $A\in \binom{[n]}{k}$, we have:
    \begin{align*}
        |\Delta_f(A) - F_D(A)| &\leq 
        |\Delta_f(A) -f(A)| + |f(A) - F_{D'}(A)| + |F_{D'}(A)-F_D(A)| \leq 2\epsilon.
    \end{align*}
    Therefore, $(\psi, \phi)$ is a $2\epsilon$-additive sketch for $\mathcal{F}_{n,k}$ and therefore by \Cref{thm:set-intersection-sketching-lower-bound} its bit complexity is $\kappa=\Omega(n \cdot k \cdot \log(\frac{n}{k})/\epsilon^2)$.
 \end{proof}

 Note that by setting $k=\Theta(n)$, we have that sketching $\mathcal{C}_{n,m}$, for $m=\Theta(n/\epsilon^2)$, requires $\Omega(n^2/\epsilon^2)$ bits. 

Moreover, a simple application of the Hoeffding bound, shows that our lower bound is tight for additive errors, a result which we now show.
\begin{restatable}{theorem}{CoverageFunctionSketchingUBAdditive}\label{thm:coverage-function-sketching-ub-additive}
For every $\epsilon > 0$, $k\geq 2$, and every normalized coverage function $f$, there exists a normalized coverage function $f'$ defined on a universe of cardinality $t = O\left(n/\epsilon^{2}\right)$ (resp., $t=O\left(\log(\binom{n}{k})/\epsilon^2\right)$) such that for each $A\subseteq [n]$ (resp., $A\in\binom{[n]}{k}$): 
\[
    |f(A) - f'(A)| \leq \varepsilon.
\]
Thus, there is an $\epsilon$-additive sketch for $\mathcal{C}_{n,m}$ (resp., $\mathcal{C}_{n,m,k}$) with bit complexity $O(n^2/\epsilon^2)$ (resp., $O\left(n \cdot \log(\binom{n}{k})/\epsilon^2\right)$).
\end{restatable}
\begin{proof}

If we are finding a sketch for $\mathcal{C}_{n,m}$, fix $t = \left\lceil \frac{(n+2) \ln 2}{2\epsilon^2} \right\rceil$, otherwise fix $t=\left\lceil \frac{\ln (4\binom{n}{k})}{2\epsilon^2} \right\rceil$.  If $m \le t$, the statement trivially follows since the total space to represent the original instance is $O(mn)=O(n\cdot t)$.  Therefore, we assume $m > t$.

We construct a sub-universe $T \subseteq U$ by selecting $t$ elements from $U$ uniformly at random, without replacement. 
Now, consider the (random) normalized coverage function $f'$ associated with the collection $S_1 \cap T, \ldots, S_n \cap T$ and the universe $T$, defined as
\[
f'(A) = \frac{1}{t} \cdot \left|\bigcup_{i \in A} S_{i} \cap T\right|.
\]
For each $A \subseteq [n]$, we have
$\E\left[f'(A)\right] = f(A)$, where the expectation is over the choice of $T$.  Applying a Hoeffding bound for sampling without replacement (see, e.g., \citep{s74}), we can bound the deviation from the mean as:
\[
\Pr[|f'(A) - f(A)| \ge \epsilon] 
\leq 2 \exp\left(-2 \epsilon^2 t\right).
\]
Now, if we are finding a sketch for $\mathcal{C}_{n,m}$, we apply a union bound over all $A \subseteq [n]$, and obtain:
\[
\Pr\left[\exists A \subseteq [n]: \left| f'(A) - f(A)\right| \ge \epsilon \right] \leq 2^n \cdot 2\exp(-2\epsilon^2t) \le2^n \cdot 2^{-(n+1)} = \frac12,
\]
by our choice of $t$ and since $n \ge 1$. Thus, $\left|f'(A) - f(A)\right|< \epsilon$ for each $A\subseteq[n]$, with probability at least $1/2$. Similarly, to obtain a sketch for $\mathcal{C}_{n,m,k}$, we apply a union bound over all $A\in \binom{[n]}{k}$ and obtain:
\[
\Pr\left[\exists A \in \binom{[n]}{k}: \left| f'(A) - f(A)\right| \ge \epsilon \right] \leq \binom{n}{k} \cdot 2\exp(-2\epsilon^2t) \le \binom{n}{k} \cdot \frac{1}{2}\cdot \binom{n}{k}^{-1} = \frac12,
\]
by our choice of $t$.

To store $f'$, notice that we only need to store $S_1 \cap T, \ldots, S_n \cap T$.  Since $|T| = t$, this uses space $t \cdot n$ bits, concluding the proof.  
\end{proof}

\subsection{Multiplicative-Error Sketch}\label{sec:coverage-func-multiplicative}

It is easy to see that Theorem~\ref{thm:coveragelb} also applies to 
$\epsilon$-multiplicative sketches of $\mathcal{C}_{n,m}$ or $\mathcal{C}_{n,m,k}$. Indeed, an $\epsilon$-multiplicative sketch also implies an $\epsilon$-additive sketch, since $f(A)\leq 1$ for each $A\subseteq [n]$. 

In terms of upper bounds, \cite{bdfknr12} provide an  $\epsilon$-multiplicative sketch that uses $\tilde{O}(n^3/\epsilon^2)$ bits.\footnote{\cite{bdfknr12} do not study the bit complexity explicitly, but provide a weighted coverage function with universe size $O(n^2/\epsilon^2)$ and weights of polynomial size in $n$ that approximates the original coverage function; therefore, a simple upper bound on their bit complexity is $\tilde{O}(n^3/\epsilon^2)$. 
} \cite{bem17} provide a multiplicative sketch for a coverage function to approximate the $k$-cover problem with bit complexity $\tilde{O}(nk)$. Their sketch is based on classical results to approximate the number of distinct elements in a data stream \citep{bjkst02,cdim02}, and can be generalized to apply to general coverage functions with bit complexity $\tilde{O}(n^2)$ when $m=n^{O(1)}$. For completeness, in Appendix \ref{app:missing-cover-func-mult-ub}, we provide a proof of the following result:

\begin{restatable}{theorem}{UBCoverageUnweightedMultiplicative}\label{thm:ub-coverage-function-unweighted-multiplicative}
Fix any $\epsilon\in(0,1)$ and $k\geq 2$. There exists an $\epsilon$-multiplicative sketch for the set of all normalized covering functions $\mathcal{C}_{n,m}$ (resp., $\mathcal{C}_{n,m,k}$) with bit complexity $O({n^2\over \varepsilon^2}\log m)$ (resp., $O(\frac{n}{\epsilon^2} \cdot \log(\binom{n}{k}) \cdot \log m)$).
\end{restatable}

When the weights are positive integers, it is possible to extend this upper bound to weighted coverage functions. Specifically, let $\mathcal{C}^W_{n,m}$ be the set of all weighted coverage functions with positive integer weights in $W=\{w_u\}_{u\in U}$, and let $\mathcal{C}^W_{n,m,k}$ be the set where we restrict the domain of the functions to $\binom{[n]}{k}$. 

\begin{restatable}{corollary}{WeightedCoverageSketch}\label{cor:weighted-coverage-function-mult}
    For any $\epsilon\in(0,1)$, $k\geq 2$, and positive integer weights $W=\{w_u\}_{u\in U}$, the set of weighted coverage functions $\mathcal{C}^W_{n,m}$ (resp., $\mathcal{C}^W_{n,m,k}$) can be $\epsilon$-multiplicatively sketched with bit complexity $O(\frac{n^2}{\epsilon^2} \cdot \log(w(U)))$ (resp., $O(\frac{n}{\epsilon^2} \cdot \log(\binom{n}{k}) \cdot \log(w(U)))$).
\end{restatable}
\begin{proof}
    Consider any $f_W\in \mathcal{C}^W_{n,m}$. We show this result by reducing $f_W$ to an unweighted coverage function $f$. 

    We build a new universe $U'$ as follows: for each $u\in U$ with positive integer weight $w(u)$, we add elements $u_1, \dots, u_{w(u)}$ to $U'$. Similarly, for each $i\in [n]$, we build $S'_i$ from $S_i$ by replacing each $u\in S_i$ with $u_1, \dots, u_{w(u)}$. Let $f$ be the normalized coverage function induced by universe $U'$ and the collection $S_1', \dots, S_n'$. Note that, for each $A\subseteq [n]$, we have:
    \[
        f_W(A) = \frac{1}{w(U)} \cdot w\left(\bigcup_{i\in A} S_i\right) = \frac{1}{|U'|} \cdot \sum_{u\in \cup_{i\in A} S_i} w(u) = \frac{1}{|U'|} \cdot \left| \bigcup_{i\in A} S'_i \right| = f(A).
    \]
    By \Cref{thm:ub-coverage-function-unweighted-multiplicative}, $\mathcal{C}_{n, |U'|}$ can be sketched with $O(\frac{n^2}{\epsilon^2} \log |U'|) = O(\frac{n^2}{\epsilon^2} \log (w(U)))$ bits which immediately implies a sketch for $\mathcal{C}^W_{n,m}$. The result for $\mathcal{C}^W_{n,m,k}$ follows similarly.
\end{proof}

Note that it is common in the literature to assume that $m=n^{O(1)}$ and that $\max_{u\in U} w(u) = n^{O(1)}$ (e.g., \citep{bdfknr12}): in this setting, the previous algorithms provide a sketch with bit complexity $O((n/\epsilon)^2 \log n)$ for $\mathcal{C}_{n,m}^W$.

\section{Application: $\ell_\infty$-Sketching of RUMs}

In this section, we describe the implication of \Cref{thm:set-intersection-sketching-lower-bound} to the problem of sketching the winning distributions of Random Utility Models (RUMs). A \emph{RUM} on $[n]$ is a probability distribution $R$ over the set of permutations on $[n]$.  
Given a non-empty $T \subseteq [n]$, we use $R_T$ to denote the
distribution of the random variable $\pi(T)$ for $\pi \sim R$, where
$\pi(T) = \argmax_{i \in T} \pi(i)$.
We use the terms \emph{slate} to denote $T$, \emph{winner} to denote $\pi(T)$, and \emph{winning distribution} to denote $R_T$. For $i\in T$, $R_T(i)$ is the probability that item $i$ is the winner among the items in $T$, according to a random permutation sampled from $R$.

Let $\mathcal{R}_n$ be the set of all RUMs over $[n]$. Here, we consider the problem of sketching the winning distributions of an arbitrary RUM in $\mathcal{R}_n$ in norm $\ell_p$, for $p \in [1, \infty)$. Specifically, for each RUM $R$ over $[n]$, we wish to construct a data structure $\hat{R}$ that, on input an arbitrary slate $T \subseteq[n]$, can provide an estimate $\hat{R}_T$ of $R_T$ satisfying: $\|{R_T - \hat{R}_T}\|_p \le \varepsilon$. \cite{chierichetti2021light} showed that $\mathcal{R}_n$ can be sketched in $\ell_1$-norm with bit complexity $O((n/\epsilon)^2 \log n)$. They also provided a lower bound of $\Omega(n^2)$ for $\ell_1$, but left open the problem for $\ell_\infty$, for which the best previously known lower bound was $\Omega(n \log n)$. Below, we use Theorem~\ref{thm:set-intersection-sketching-lower-bound} to improve this lower bound to $\Omega(n^2)$. Together with the upper bound on $\ell_1$, this shows that sketching $\mathcal{R}_n$ under any $\ell_p$-norm, $p\in [1,\infty)\cup\{\infty\}$, requires $\tilde{\Theta}(n^2)$ bits.

In fact, we prove a stronger result: we show that even if we fix a special item $i^* \in [n]$ and we are only interested in sketching $f^{i^*}_R(T):=R_{T\cup \{i^*\}}(i^*)$ for each slate $T\in \binom{[n]\setminus\{i^*\}}{k-1}$ with $k\geq 3$, then $\Omega(\frac{n \cdot k}{\epsilon^2} \log(n/k))$ bits are necessary.

\begin{theorem}
Let $3\leq k\leq n/2$, $\epsilon \in \left(\frac{1}{\sqrt{n}}, 1\right)$, and $i^*\in [n]$. Any $\epsilon$-additive sketch for the class $\{f^{i^*}_R\}_{R\in \mathcal{R}_n}$ has bit complexity $\Omega\left(\frac{n \cdot k }{\epsilon^2} \log(n/k)\right)$.
\end{theorem}
\begin{proof}
We prove this result by reducing the problem of sketching the intersection profile of a power-set distribution to that of sketching the winning distribution of a RUM. Without loss of generality we will assume that $i^* =n$ and for ease of notation we will write $f_R$ for $f_R^{i^*}$.

Suppose that there is an $\epsilon$-additive sketch $(\phi, \psi)$ for $\{f_R\}_{R\in \mathcal{R}_n}$ with bit complexity $\kappa$. Specifically, for each RUM $R\in \mathcal{R}_n$, let $\Delta_R(\cdot) := \psi(\phi(R), \cdot)$, so that $|\Delta_R(T)-f_R(T)| \leq \epsilon$ for each $T\in\binom{[n-1]}{k-1}$. Let $D$ be any power-set distribution over $[n-1]$. Consider the RUM $R$ defined as follows: sample $S\sim D$, then return a permutation $\pi$, where the top-most $|S|$ items of $\pi$ are the items of $S$ sorted in increasing order; the item in position $|S|+1$ of $\pi$ is $n$; the last $n -1 - |S|$ items of $\pi$ are the items of $[n-1] \setminus S$, also sorted in increasing order. Note that $n$ loses against all items in $S$ and wins against all items in $[n-1]\setminus S$. For any $T\subseteq [n-1]$, we have that:
\[
    f_R(T) = R_{T\cup \{n\}}(n) = \Pr_{S\sim D}[S\cap T = \varnothing] = 1 - F_D(T).
\]
Therefore, for each $T\in\binom{[n-1]}{k-1}$:
\begin{align*}
    |(1- \Delta_R(T)) - F_D(T)| &\leq |(1- f_R(T)) - F_D(T)| + |(1- \Delta_R(T)) - (1- f_R(T))|\leq \epsilon.
\end{align*}
Therefore, $(\phi, 1-\psi(\cdot, \cdot))$ is an $\epsilon$-additive sketch with bit complexity $\kappa$ for $\mathcal{F}_{n-1,k-1}$. Therefore, by \Cref{thm:set-intersection-sketching-lower-bound}, the sketch for $\{f_R\}_{R\in \mathcal{R}_n}$ must have bit complexity $\kappa=\Omega\left(\frac{(n-1)\cdot (k-1)}{\epsilon^2} \log({n-1\over k-1})\right)=\Omega\left(\frac{n\cdot k}{\epsilon^2} \log({n\over k})\right)$. 
\end{proof}

Observe that with a reduction analogous to the one used in the above proof, it is possible to show that any $\epsilon$-additive sketch for $\mathcal{F}_{n-1,k-1}$ with bit complexity $\kappa$ can be transformed into an $\epsilon$-additive sketch for $\{f_R^{i^*}\}_{R\in \mathcal{R}_n}$ for $i^*\in[n]$ with the same bit complexity $\kappa$. Hence, by \Cref{thm:upper-bound-intersection-profile}, $\{f_R^{i^*}\}_{R\in \mathcal{R}_n}$ can be $\epsilon$-additively sketched with $O(n\cdot k \cdot \log(\frac{n}{k}) / \epsilon^2)$ bits. 

Note that when we want to approximate the RUM on slates of all cardinality simultaneously, the previous result shows that $\Omega(n^2/\epsilon^2)$ bits are necessary. 

\section{Conclusions and Future Work}
In this paper we studied the bit complexity of sketching (i) vertex boundaries in graphs, (ii) coverage functions, and (iii) Random Utility Models. Our results are applications of sketching results for intersection profiles of power-set distributions, which can equivalently be regarded as results on sketching itemset frequencies in a database.  %

For the problem of sketching all the winning distribution of a RUM in $\ell_\infty$-distance, our results leave a gap of $\log n$. Is it possible to improve the lower bound or does there exist a more efficient sketch? %

The work of \cite{bdfknr12} provides a more expensive sketch for coverage functions, but their  sketch is \emph{proper}, in that it is itself  a coverage function. Do there exist proper sketches of size $\tilde{O}(n^2)$ bits when approximating all the query sets?

\section*{Acknowledgments}

We thank Justin Thaler for bringing to our attention the relationship between sketching intersection profiles and sketching itemset frequencies, which subsumed our lower bound.  

We also thank Bobby Kleinberg and Robert Krauthgamer for useful discussions on coverage functions and vertex neighborhoods. 

\bibliographystyle{plainnat}
\bibliography{main}

\newpage
\appendix

\section{Missing Proofs from \Cref{sec:coverage-func-multiplicative}}\label{app:missing-cover-func-mult-ub}
In this section, we show an algorithm that multiplicatively sketches $\mathcal{C}_{n,m}$ with $O((n/\epsilon)^2 \log m)$ bits. 
We start by introducing some notation. For a set $S \subseteq \mathbb{N}$ (i.e., without duplicates) and integers $k\geq 1$, and $M\geq 1$, let:
\begin{align*}
    \min{}^{(k)}(S) &= \begin{cases}
        S & \text{if}\quad |S|< k,\\
        \text{$k$ minimum distinct values of $S$}&\text{if}\quad |S|\geq k.
    \end{cases}\\
    \hat{F}_0^{(k,M)}(S) &= \begin{cases}
        |S| &\text{if}\quad |S|< k\\
        \frac{k\cdot M}{\max(S)} &\text{if}\quad |S|\geq k
    \end{cases}
\end{align*}
We will leverage on the well-known $k$-minimum value (KMV) sketch of \cite{bjkst02}, which provides the following guarantees:
\begin{theorem}[Theorem 1 of \citet{bjkst02}, paraphrased]\label{thm:f0-sketch}
    Let $\epsilon, \delta\in (0,1)$, and let $S$ be a multiset over the universe $[m]$, with $|S|=n$. Let $h:[m] \rightarrow [M]$, where $M=m^3$, be a uniformly random function. For $k=O(1/\epsilon^2)$, let 
    \[
        \hat{S}= \min{}^{(k)}\left(\{h(s) \mid s \in S\}\right).
    \]
    Then, with probability $>1/2$, it holds that:
    \[
        (1-\epsilon) \cdot F_0(S) \leq \hat{F}_0^{(k,M)}(\hat{S}) \leq (1+\epsilon) \cdot F_0(S),
    \]
    where $F_0(S)$ is the number of distinct elements of $S$. Storing $\hat{S}$ requires $O(\frac{\log m}{\epsilon^2})$ bits. 

    Moreover, suppose to run the algorithm $t=O(\log (1/\delta))$ times independently, and let $\hat{S}_1, \dots, \hat{S}_t$ be the corresponding sets produced. Let $v$ be the median value of $\hat{F}_0^{(k,M)}(\hat{S}_1), \dots, \hat{F}_0^{(k,M)}(\hat{S}_t)$. Then, with probability $\geq 1-\delta$, it holds that:
    \[
        (1-\epsilon) \cdot F_0(S)\leq v \leq (1+\epsilon) \cdot F_0(S).
    \]
\end{theorem}

Note that the original sketch of \cite{bjkst02} was developed for streaming algorithms and uses a family of pairwise-independent hash functions rather than sampling $h$ uniformly at random. However, we do not need to store $h$ for our purposes, and we can therefore just pick a uniform at random function for simplicity. 

A key property of this sketch is that it is easily composable. In particular, given the sketches for multisets $S_1, \dots , S_n$ one can compute a sketch for their (multi)union $S = S_1\cup\dots \cup S_n$, that will be statistically equivalent to a sketch generated from scratch for $S$. We remark that \cite{bem17} used exactly the same idea using the sketch of \cite{cdim02} in their Appendix D, but limited to sets of size $k$. For completeness, we provide here a full proof. 

Most crucially for our application, storing all sketches for a collection $S_1, \dots , S_n$ one can compute the sketch for the set:
\[
    S_A = \bigcup_{a\in A} S_a
\]
for \emph{all} subsets $A\subseteq [n]$. The guarantees of the basic sketch yield that \emph{for each} choice of $A$, with probability at least $1-\delta'$ the sketch for $S_A$ will have small multiplicative error. Setting $\delta' = \Theta(\delta{2^{-n}})$ allows us to guarantee that, with probability at least $1-\delta$, \emph{for all} choices of $A$ the sketch will have smaller multiplicative error.

\UBCoverageUnweightedMultiplicative*
\begin{proof}
Consider any normalized coverage function $f\in \mathcal{C}_{n,m}$. We show how to construct a data structure that approximates $f(A)$ within a multiplicative error of $(1\pm \epsilon)$ for any $A\subseteq[n]$.

Let $\delta=2^{-(n+1)}$, $M=m^3$, and let $h_1, \dots, h_t$ be functions from $[m]$ to $[M]$ sampled i.i.d. uniformly at random, where $t=\Theta(\log(1/\delta))=\Theta(n)$ as required by \Cref{thm:f0-sketch}. Let $k=\Theta(1/\epsilon^2)$ as required by \Cref{thm:f0-sketch}. If the coverage function $f$ consists of a collection of sets $S_1, \dots, S_n\subseteq U$, then, we store, for each $i\in [n]$, and $j\in [t]$ the set:
\[
    \hat{S}_i^{(j)} = \min{}^{(k)}(\{h_j(s) \mid s \in S_i\}).
\]
Note that we do not store the functions $h_1, \dots, h_t$. Therefore, the total bit complexity of our data structure is at most: $n \cdot t \cdot k \cdot \log_2(M) = O( (n^2 / \epsilon^2) \cdot \log m )$.

Upon receiving a query $A \subseteq [n]$, we return the estimate:
\[
    \hat{f}(A) = \operatorname{median} \left\{ \hat{F}_0^{(k,M)}\left(\min{}^{(k)}\left(\bigcup_{a \in A}\hat{S}_{a}^{(j)}\right)\right)\right\}_{j\in [t]}.
\]
Observe that, for each $j\in [t]$, we have:
\[
    \min{}^{(k)}\left(\bigcup_{a \in A}\hat{S}_{aj}\right) = \min{}^{(k)}\left(\bigcup_{a \in A} \{h_j(s) \mid s \in S_a \}\right),
\]
therefore, $\hat{f}(A)$ has the same distribution as the estimate computed from scratch by the algorithm of \Cref{thm:f0-sketch} for the multiset $\cup_{a\in A} S_a$. Hence, with probability at least $1-\delta$:
\[
    (1-\varepsilon)\cdot m \cdot f(A) \le \hat{f}(A)\le (1+\varepsilon)\cdot m \cdot f(A),
\]
where we recall that $f(A)=\frac{1}{m}|\bigcup_{a\in A}S_a|$. In particular, by the union bound, with probability at least $1/2$ we have:
\[
    \forall A \subseteq [n] : \quad (1-\varepsilon)\cdot m \cdot f(A) \le \hat{f}(A)\le (1+\varepsilon)\cdot m \cdot f(A), 
\]
and the sketch will be correct as needed. 

Note that to sketch $\mathcal{C}_{n,m,\ell}$, it suffices to set $\delta=\frac{1}{2}\binom{n}{\ell}^{-1}$ and the rest follows as before.
\end{proof}

\section{Typical Set of Intersection Profiles}\label{sec:typical-set}
In this section, we substantiate the claim that sketching the intersection profile of most power-set distributions requires little memory.

Consider:
\[
    \mathcal{T}^{(n)}_{\varepsilon} := \left\{D \in \mathcal{D}_n \,\middle|\, \text{for all } S \in 2^{[n]}, {1\over 2^{|S|}}- \varepsilon \le 1-F_D(S)\le  {1\over 2^{|S|}} + \varepsilon\right\}.
\]
We say that a distribution $D \in \mathcal{D}_n$ is $\varepsilon$-\emph{typical} if it  belongs to $\mathcal{T}^{(n)}_\varepsilon$ (the $\varepsilon$-\emph{typical set}). We show that as $n\to\infty$ the probability that a uniformly random distribution $D$ is $\varepsilon$-typical tends to $1$.

\begin{theorem}\label{thm:typical-set-is-typical}
    For any $\varepsilon >0$:
    \begin{equation}\label{eq:typical-set-is-typical}
        \lim_{n\to \infty}\Pr_{D\sim U(D_n)} \left[ D \in \mathcal{T}_\varepsilon^{(n)}\right] = 1,
    \end{equation}
    where $U(D_n)$ is the uniform distribution over the $(2^{n}-1)$-dimensional simplex in $\mathbb{R}^{2^{n}}$. 
\end{theorem}

\noindent
In fact, we will prove the following stronger, non-asymptotic version of \Cref{thm:typical-set-is-typical}.

\begin{theorem}\label{lem:non-asymptotic-typicality}
    For any $\varepsilon >0$, and any $n \ge 20\log_{2}{1\over \varepsilon}$:
    \[ 
        \Pr_{D\sim U(D_n)} \left[ D \in \mathcal{T}_\varepsilon^{(n)}\right] \ge {1- {2\over e^{2^{0.27 n}}}}.
    \]
\end{theorem}
\begin{proof}
Recall that the pdf of a Dirichlet distribution with parameters $(a_1, \ldots, a_K)$ is given by
\[
    f(x_1, \ldots , x_K) = {\Gamma\left(\sum_{i=1}^Ka_i\right) \over \prod_{i=1}^K\Gamma(a_i)} \prod_{i=1}^K x_i^{a_i -1},
\]
where $\Gamma(\cdot)$ is the gamma function and $\sum_{i=1}^Kx_i = 1$ and $x_i\geq 0$ for each $i$.

We note that sampling a distribution $D$ from $U(D_n)$ is equivalent to sampling a vector $X = (X_A)_{A\subseteq [n]}$ from a flat Dirichlet distribution with $2^n$ parameters $(1, \dots ,1)$, where each $X_A$ indicates the probability of selecting set $A$ from $D$. 

Fix a subset $S \subseteq [n]$. We define $\tilde{X}_S :=\sum_{A\subseteq[n]\,:\, A\cap S \neq \varnothing} X_A$ and $\tilde{Y}_S:= \sum_{A\subseteq[n]\,:\, A\cap S = \varnothing} X_A$.

We recall the following standard result (see, e.g.,\ \cite{f73}):
\begin{proposition}[Aggregation Property of Dirichlet Random Variables]
    Let $X = (X_1, \dots , X_K) \sim \Dirichlet(a_1, \ldots,a_K)$. For any $i< j \in[K]$, let $X^{(i\leftrightarrow j)}$ be the vector obtained from $X$ by removing $X_j$ and replacing $X_i$ with $X_i + X_j$, we then have:
    \[
        X^{(i\leftrightarrow j)} \sim \Dirichlet (a_1, \dots , a_{i-1},a_i+a_j, a_{i+1}, \dots, a_{j-1}, a_{j+1}, \ldots, a_K).
    \]
\end{proposition}

By the aggregation property of the Dirichlet distribution, we have:
\[
    \left(\tilde{X}_S, \tilde{Y}_S\right) \sim \Dirichlet (\alpha,\beta),
\]
where $\alpha = |\{A\in2^{[n]}\mid A\cap S \neq \varnothing\}|$ and $\beta = |\{A\in2^{[n]}\mid A\cap S = \varnothing\}|$. The marginal of a Dirichlet distribution follows a beta distribution (again, see, e.g.,\ \cite{f73}):
\[
   \tilde{X}_S \sim \Beta\left(\alpha,\beta\right).
\]
The latter is sub-Gaussian with a proxy-variance upper bounded by $\widetilde{\sigma}^{2} = {1\over 4(\alpha + \beta + 1)} \le {1\over 2^{n+2}}$~\citep[Theorem 2.1]{ma17}, meaning it satisfies:
\[
    \Pr[\tilde{X}_S - \mu_S \ge \varepsilon] \le e^{-{\varepsilon^2\over 2\widetilde{\sigma}^2}} \le e^{-\varepsilon^2 2^{n+1}},
\]
where $\mu_S := 1-{1\over 2^{|S|}}$ is the mean of $\tilde{X}_S$, since $\E{}{[\tilde{X}_S]}=\frac{\alpha}{\alpha+\beta}$.
In particular, if $n \ge 20\log_2 {1\over \varepsilon}$, we have:
\[
    \Pr[\tilde{X}_S - \mu_S \ge \varepsilon] \le \operatorname{exp}\left({-2^{0.9n+1}}\right) \le \operatorname{exp}\left({-n\ln(2) - 2^{ 0.27n}}\right) = 2^{-n} e^{-2^{0.27 n}}.
\]
Applying a union bound, gives:
\[
    \Pr[\exists S \subseteq [n] : \tilde{X}_S - \mu_S \geq \varepsilon ] \le e^{-2^{0.27n}}.
\]
By considering the variables $\{\tilde{Y}_S\}_{S\subseteq [n]}$ and applying the same argument as above, since we have $\tilde{Y}_S=1-\tilde{X}_S$, we obtain:
\[
    \Pr[\exists S \subseteq [n] : \mu_S - \tilde{X}_S \geq \varepsilon ] \le e^{-2^{0.27n}},
\]
and hence:
\[
    \Pr[\exists S \subseteq [n] : |\tilde{X}_S - \mu_S| \geq \varepsilon ] \le 2e^{-2^{0.27n}}.
\]
Noting that $\tilde{X}_S = F_D(S)$ completes the proof of \Cref{lem:non-asymptotic-typicality} and hence that of \Cref{thm:typical-set-is-typical}.
\end{proof}

\end{document}